\documentclass[sn-mathphys,Numbered]{sn-jnlmod}

\usepackage{dutchcal}
\usepackage{bm}

  \usepackage{paralist}

\usepackage{graphicx}%
\usepackage{multirow}%
\usepackage{amsmath,amssymb,amsfonts}%
\usepackage{amsthm}%
\usepackage{mathrsfs}%
\usepackage[title]{appendix}%
\usepackage{xcolor}%
\usepackage{textcomp}%
\usepackage{manyfoot}%
\usepackage{booktabs}%
\usepackage{algorithm}%
\usepackage{algorithmicx}%
\usepackage{algpseudocode}%
\usepackage{listings}%

\theoremstyle{thmstyleone}%
\newtheorem{theorem}{Theorem}
\newtheorem{proposition}[theorem]{Proposition}%

\theoremstyle{thmstyletwo}%
\newtheorem{remark}{Remark}%

\theoremstyle{thmstylethree}%
\newtheorem{definition}{Definition}%

\begin{document}

\title[The theory of screws derived from a module]
 {The theory of screws derived from a module \\ over the dual numbers}

\author{\fnm{Ettore} \sur{Minguzzi}}\email{ettore.minguzzi@unifi.it}

\affil{\orgdiv{Dipartimento di Matematica e Informatica ``U. Dini''}, \orgname{Universit\`a
degli Studi di Firenze}, \orgaddress{\street{Via S. Marta 3}, \city{Firenze}, \postcode{I-50139},  \country{Italy}}}


\abstract{The theory of screws clarifies many analogies between apparently unrelated notions in mechanics, including the duality between forces and angular velocities.
It is known that the real 6-dimensional space of screws can be endowed with an operator $\mathcal{E}$, $\mathcal{E}^2=0$, that converts it into a rank 3 free module over the dual numbers. In this paper we prove the converse, namely, given a rank 3 free module over the dual numbers, endowed with orientation and a suitable scalar product ($\mathbb{D}$-module geometry), we show that it is possible to define, in a canonical way, a Euclidean space so that each element of the module is represented by a screw vector field over it. The new approach has the effectiveness of motor calculus while being independent of any reduction point. It gives insights into the transference principle by showing that affine space geometry is basically vector space geometry over the dual numbers. The main results of screw theory are then recovered by using this point of view.}

\maketitle
\tableofcontents

\section{Introduction}

The goal of this work is to provide a new type of introduction to the mathematics of screw theory, one that, we believe, has the merit of elucidating its foundations while preserving, and actually improving, effectiveness in calculations. Our formulation will be  based on $\mathbb{D}$-module geometry, where $\mathbb{D}$ is the ring of dual numbers.
Applications of this approach to mechanics as well as  geometry will be given in a next work.

Screw theory has a venerable history. It originated from the work of Sir R. Ball \cite{ball00} as a method for unifying several geometrical notions from  classical mechanics and received, in its early development at the end of the nineteenth century and beginning of twentieth century, contributions from eminent mathematicians including Clifford \cite{clifford73}, Study \cite{study03}, Kotel'nikov \cite{kotelnikov95} and von Mises \cite{vonmises24a,vonmises24b,wohlhart95,wohlhart02}. It was revived in the last decades for its powerfulness in dealing with applications such as robotics, particularly via von Mises motor (screw) calculus \cite{brand47,dimentberg68}.

Motor calculus makes use of dual numbers, dual vectors and dual matrices. For an introduction to motor calculus see \cite{veldkamp76,bottema79,mccarthy90,fischer99}, while for interesting applications to mechanics including use of dual Euler angles, the reader is referred to \cite{yang64,ying02,pennock02,wittenburg16}. For a  study of the spectral properties of dual matrices, the reader is referred to the recent work by Gutin \cite{gutin21}.

We can recognize two main types, essentially equivalent, forms of presentation of screw theory, only differing in some minor terminological aspects. Both share the same starting point, that  is, the definition of the ambient space, namely Euclidean space\\
\begin{itemize}
\item[($\star$)] The {\em Euclidean space} $E$ is an affine space modelled over $(V,\cdot, \textrm{or})$, where $V$ is a real vector space of dimension 3 endowed with a scalar product  and an orientation.\\
\end{itemize}
The space $V$ is also called the {\em vector space of free vectors}.

 A pair $(P,(b_1,b_2,b_3))$ where $(b_1,b_2,b_3)$ is a basis for $V$ and $P\in E$, is said to be a {\em reference frame} for $E$.

The first approach alluded to above has been  presented in my previous work \cite{minguzzi12} to which the reader is referred for the kinematical and mechanical interpretations omitted in this work. It  can be conveniently introduced starting from  the following important result \cite{delassus13,ramis90}:\\

\begin{theorem}[Delassus]
For a vector field $\mathcal{s}\colon E\to V$, $P\mapsto \mathcal{s}(P)$, the  following conditions are equivalent
\begin{itemize}
\item[(a)]  for every $P,Q\in E$
\begin{equation}
 \mathcal{s}(P)\cdot (Q-P)=\mathcal{s}(Q)\cdot (Q-P), \qquad \textrm{(equiprojective condition)}
\end{equation}
\item[(b)]  there is $s\in V$ such that for every $P,Q\in E$
\begin{equation} \label{bid}
\mathcal{s}(Q)-\mathcal{s}(P)=s\times (Q-P). \qquad \textrm{(constitutive equation)}
\end{equation}
\end{itemize}
\end{theorem} $\empty$

A field that satisfies any of these equivalent conditions is called a {\em screw}  and Eq.\ (\ref{bid}) is then called the {\em constitutive equation of the screw}.

We mention that there are two major physical realizations
of screws, namely twists and wrenches, for which more details can be found, e.g., in \cite{minguzzi12}. Twists describe, at a given time, the velocity field of points belonging to a rigid body in motion. Wrenches describe the field of mechanical momentum as the pole is varied over Euclidean space.
The current document focuses on a geometric-algebraic approach where these two physical interpretations are not of  primary interest.


It is easy to prove that $s$ in (b) is unique and hence we call it {\em resultant} of the screw. The screws form a 6-dimensional real  vector space $\mathcal{S}$, and the map $\mathcal{R}\colon \mathcal{S}\to V$, $\mathcal{s}\to s$, which associates to a screw its resultant, is linear.

This approach to screw theory via special vector fields has the advantage of being independent of any reduction point. It goes back, at least, to Lovell III's PhD thesis \cite{lovell86} and has long been used in French engineering universities where screws are called {\em torseurs} \cite{ramis90,chevallier91,pommier10,minguzzi12,chevallier18}. We stress that in this work the screws will always be elements of a 6-dimensional real  vector space, hence our notion of screw is equivalent to that of torseur, and conforms with the  terminology used by Dimentberg \cite{dimentberg68}. This terminology differs from the original terminology used by Sir R. Ball.


A screw with zero resultant is  a constant vector field. A screw with non vanishing resultant is also called {\em proper screw}. It   is invariant under translations along the resultant $\mathcal{s}(P+\lambda s)=\mathcal{s}(P)$, for each $\lambda \in \mathbb{R}$,  and is such that the locus $\{P:  \mathcal{s}(P)\propto s\}$ is a line called {\em the screw axis} (the symbol $\propto$ means `proportional to'). The value of the screw on the axis is called {\em vector invariant} $\mathcal{s}_\parallel$, and the number $p\in \mathbb{R}$ such that $\mathcal{s}_\parallel = p s$ is called {\em pitch}. Denoting with $O$ any point on the axis, we have from Eq.\ (\ref{bid})
\begin{equation}
\mathcal{s}(P)=\mathcal{s}_\parallel+s\times (P-O) ,
\end{equation}
namely the screw is indeed a field with a screw development over $E$, with a `translational' component $\mathcal{s}_\parallel$ and a `rotational' component $s\times (P-O)$. Pictorially, a screw can be represented as two vectors, namely $s$ and $\mathcal{s}_\parallel$ aligned on the screw axis as from this data the screw can be recovered.

Clearly, via Eq.\ (\ref{bid}) a screw can be recovered from the knowledge of the pair $(s, \mathcal{s}(P))$ where $P$ is any point of $E$. For any $P\in E$ the pair  $(s, \mathcal{s}(P))$ is called {\em motor} and the map $\mathcal{S}\to V\times V$,  $\mathcal{s}\mapsto (s, \mathcal{s}(P))$, is called {\em motor reduction} at $P$ (and $P$ is also referred as the {\em reduction point}) \cite{dimentberg68}. A choice of basis for $V$ provides us with six real numbers called {\em Pl\"ucker coordinates} of the screw. The motor approach is the second approach mentioned above. In it all expressions refer to the chosen reduction point. The independence of some constructions of the reduction point is less obvious but, as a formalism,  it is more closely related to motor calculus that we shall introduce in a moment.

Given two screws, we have two useful operations on them. Firstly, a {\em screw scalar product (comoment)} $\langle,\rangle:\mathcal{S}\times \mathcal{S}\to \mathbb{R}$ defined through
\begin{equation}
\langle \mathcal{s}_1, \mathcal{s}_2\rangle:= s_1\cdot \mathcal{s}_2(P)+ \mathcal{s}_1(P) \cdot s_2
\end{equation}
where, thanks to Eq.\ (\ref{bid}), the right-hand side is actually independent of $P\in E$, and secondly, the {\em commutator} of screws $[,]: \mathcal{S}\times \mathcal{S}\to \mathcal{S} $ which is the negative (of the) Lie bracket
\begin{equation}
[\mathcal{s}_1,\mathcal{s}_2](P):=-[\mathcal{s}_1,\mathcal{s}_2]_{L}(P)=s_1\times  \mathcal{s}_2(P)+\mathcal{s}_1(P) \times s_2.
\end{equation}
Again, by using Eq.\ (\ref{bid}) it is easy to check that it is a screw whose resultant is $s_1\times s_2$.

Oriented lines on $E$ are just particular screws, sometimes called {\em spears} or {\em unit line vectors} \cite{huang13}. Let $O$ be a point on the line and let $e\in V$ be a unit vector giving the orientation. The screw $\ell(P):=e \times (P-O)$ is independent of $O$ and can be used to represent the oriented line on $E$.

We have the formula
\begin{align}
\langle \ell_1,\ell_2\rangle=-d \sin \theta
\end{align}
where $\theta$ is the angle between $e_1$ and $e_2$ and $d$ is the distance between the lines, and the formula
\begin{equation}
[\ell_1,\ell_2](P)=\cos \theta (B-A)+(e_1\times e_2)\times (P-A)
\end{equation}
where $A$, $B$ are closest points, with $A$ in the axis of $\ell_1$ and $B$ in the axis of $\ell_2$.

It is now very useful to introduce the commutative ring of dual numbers $\mathbb{D}=\mathbb{R}[\epsilon]/(\epsilon^2)$. Introducing the symbol $\epsilon$ and the rule $\epsilon^2=0$, we can identify $\mathbb{D}$ with the {\em dual numbers}
\[
\mathbb{D}=\{a+\epsilon b: \ a,b\in \mathbb{R}\}=\mathbb{R}+\epsilon \mathbb{R}.
\]
Given a dual number $x=a+\epsilon b$ we define $\mathfrak{Re}(x)=a$ and $\mathfrak{Du}(x)=b$ and call the former the real part, and the latter the dual part. The {\em conjugate} of $x=a+\epsilon b$ is $x^*=a-\epsilon b$.
The inverse (necessarily unique) exists only for $\mathfrak{Re}(x)=a\ne 0$ and is given by $x^{-1}=\frac{x^*}{x^* x}=\frac{a-\epsilon b}{a^2}$.

If $f:\mathbb{R}\to\mathbb{R}$ is a real analytic function we define its extension to the dual numbers as follows  $f(a+\epsilon b):= f(a)+\epsilon b f'(a)$, i.e.\ compatibly with the Taylor expansion and the rule $\epsilon^2=0$. As the power series representation over the dual numbers does not change,
the properties of $f$ resulting from this series hold true for the dual extension as
well. As an example of extension, if $x=a+\epsilon b$ is such that $a>0$
\[
\sqrt{x}=\sqrt{a}+\epsilon \frac{b}{2\sqrt{a}},
\]
so that $x=(\sqrt{x})^2$.

Von Mises motor calculus consists in choosing a reduction point $O\in E$ and considering the rank 3 free module over the dual numbers $V\oplus \epsilon V$ whose elements are the  motors written in the form $m=s +\epsilon \mathcal{s}(O)$.  The map $\phi_O: \mathcal{S}\to V\oplus \epsilon V$ can also be called {\em motor reduction}.

We recall that a module is the closest mathematical object to a vector space when the scalars do not belong to a field but rather to a ring \cite[Sec.\ 3.1]{lang02}. Elements $\{v_i\}$ of a module, finite in number, are said to be linearly independent if $\sum_i\lambda_i v_i=0 \Rightarrow \forall i, \lambda_i=0$, where the coefficients $\lambda_i$ belong to the ring (since the coefficients $\lambda_i$ are not necessarily invertible,  linear dependence  {\em does not} imply that one element can be expressed as a linear combination of the others). The definition of {\em span} of some finite set of elements of the module is analogous to that for vector spaces, as in the notion of {\em basis}. The module is {\em free} if it admits a basis \cite[Sec.\ 3.4]{lang02}, and its {\em rank} is the cardinality of the basis, that is, the rank is the analog of  the dimension for a vector space.

Extending the scalar product and cross product to $V\oplus \epsilon V$  in the obvious way,  we get
\begin{align}
m_1\cdot m_2&= s_1\cdot s_2+ \epsilon \langle \mathcal{s}_1, \mathcal{s}_2 \rangle, \label{vxe}\\
m_1\times m_2&=s_1\times s_2+\epsilon [\mathcal{s}_1,\mathcal{s}_2](O) ,  \label{vxr}
\end{align}
where the latter equation tells us that the motor reduction of the commutator is the cross product of the motor reductions.

For any  two oriented lines $\ell_1$ and $\ell_2$ let us introduce  Study's dual angle $\Theta=\theta+\epsilon d$, which gives information on both the relative orientation and distance of the lines. The nice fact is that for lines the formulas (\ref{vxe})-(\ref{vxr}) can be written
\begin{align}
\ell_1\cdot \ell_2&= \cos \Theta, \label{one} \\
\ell_1\times \ell_2&=\sin \Theta \, \ell_3 , \label{two}
\end{align}
where $\ell_3$ is the oriented line of direction $e_1\times e_2$ passing through the points of closest distance $A$ and $B$. These formulas are perfectly analogous to the formulas for scalar and cross product in vector geometry but retain information on the `affinity' of the geometric objects, namely on their position in affine space. For instance, two lines are orthogonal when their directions are orthogonal (i.e.\ $e_1\cdot e_2=0$) and they intersect.

Motor calculus has proved particularly convenient in robotics where we can consider a basis  of applied orthonormal vectors on each rigid body of the system. In terms of motors the condition of orthonormality reads $m_i\cdot m_j=\delta_{ij}$, but now the motors $m_i$ are not free vectors and  hence retain all the information on the origin of the frame. A change in the applied frame $m'_i=U_{ij} m_j$ is then mediated by an orthogonal matrix of dual numbers: $U^T U=I$. The matrix does not depend on the reduction points because  $U_{ij}=m'_i\cdot m_j$ which by Eq.\ (\ref{vxe}) depends on the two screws with motors $m'_i$, $m_j$ but is independent of any reduction point.
 The kinematics of the robot becomes ultimately described by a composition of $3\times 3$ dual matrices.

Motor calculus shows how screws can be considered as elements of a rank 3 module where, however, the whole motor representation depends on the chosen reduction point, so that the independence of some constructions of the reduction point is not so clear at first sight. A very interesting approach was introduced by Chevallier \cite{chevallier91} who was able to remedy this problem. The idea is to introduce the map  (cf.\ \cite{chevallier91}, operation $(V)$)
\[
\mathcal{E}\colon \mathcal{S}\to \mathcal{S}, \quad \mathcal{s}\mapsto s,
\]
which associates to each screw a constant screw given at each point by the resultant. Clearly, every constant screw has zero resultant thus $\mathcal{E}^2=0$. Chevallier  shows that the 6-dimensional real vector space $\mathcal{S}$ can be seen as a rank 3 free module over $\mathbb{D}$ by defining (cf.\ \cite{chevallier91}, equation (18))
\begin{equation} \label{jgg}
\epsilon \mathcal{s}:=\mathcal{E}(\mathcal{s}).
\end{equation}
In fact, thanks to this definition, it now makes sense to take linear combinations with dual number coefficients of screws, e.g.\ $(a+\epsilon b)  \mathcal{s}_1+ (c+d \epsilon) \mathcal{s}_2$ as this operation, giving the sum of four screws, returns a screw  (for more details we refer to \cite[Theorem 1]{chevallier91}).
Observe that the motor reduction map satisfies  $\phi_O(\mathcal{E}(\mathcal{s}))=\epsilon  \phi_O(\mathcal{s})$, thus $\phi_O$ provides an isomorphism of $\mathbb{D}$-modules.

The module is free \cite{lang02} because given an orthonormal positively oriented frame $(P, (e_1,e_2,e_3))$ of $E$, the triple of orthogonal oriented lines $(\ell_1,\ell_2,\ell_3)$ provides a basis for the module $\mathcal{S}$, that is, any screw  $\mathcal{s}\in \mathcal{S}$ can be written in a unique way in the form
\[
\mathcal{s}=\sum_{i=1}^3 (a_i+b_i \epsilon) \, e_i
\]
where $(a_i+b_i \epsilon)\in \mathbb{D}$ (here $(a_1,a_2,a_3)$ are the components of the resultant $s$ of $\mathcal{s}$ while $(b_1,b_2,b_3)$ are the components of $\mathcal{s}(P)$).
We can also induce an orientation on the module $\mathcal{S}$ by specifying which bases are {\em positive} as follows

\begin{quote} \label{jgt} {\normalsize
Let  $(e_1,e_2,e_3)$  be a positive orthonormal reference frame for $(E,\cdot, $ $ or)$, then we say that on $\mathcal{S}$ the basis
$(\ell_1, \ell_2, \ell_2)$ is {\em positive} together with all the other bases connected  to it by  dual matrices whose real part has positive determinant.}\\
\end{quote}

This definition of {\em positive basis} on $\mathcal{S}$ is well posed as independent of the starting reference frame.

Chevallier also observes that if we define $\circ\colon \mathcal{S}\times \mathcal{S} \to \mathbb{D}$ through
\begin{equation} \label{jgy}
\mathcal{s}_1\circ \mathcal{s}_2=s_1\cdot s_2 + \epsilon\, \langle \mathcal{s}_1, \mathcal{s}_2\rangle
\end{equation}
then $\mathcal{E}(\mathcal{s}_1)\circ \mathcal{s}_2=\epsilon \mathcal{s}_1\circ \mathcal{s}_2$, namely $\circ$ is $\mathbb{D}$-bilinear. Thus $\mathcal{S}$ can be seen as a rank 3 free module over $\mathbb{D}$ endowed with an orientation and a $\mathbb{D}$-bilinear {\em scalar product}.

This scalar product  terminology should be clarified (see also Def.\ \ref{bas}).
Observe that $\mathcal{E}(\mathcal{s})\circ \mathcal{E}(\mathcal{s})=\epsilon^2 \mathcal{s}\circ \mathcal{s}=0$ thus the elements in $\mathcal{E}(\mathcal{S})$ have vanishing $\circ$-square. By scalar product we mean that for every $\mathcal{s}\in \mathcal{S}$, $\mathfrak{Re}(\mathcal{s}\circ \mathcal{s}) \ge 0$, with equality holding iff $\mathcal{s}\in \mathcal{E}(\mathcal{S})$. This is indeed the case because $\mathfrak{Re}(\mathcal{s}\circ \mathcal{s})=s \cdot s$, and since $\cdot$ is a scalar product, $s=0$, which means indeed $\mathcal{s}\in \mathcal{E}(\mathcal{S})$.

Chevallier also obtains other results connected to the Lie algebra structure of $\mathcal{S}$ induced by the cross product \cite{chevallier91}. These results, that are here not mentioned, have a less prominent role in our treatment as they will be deduced from the main structure with which we shall work. In summary, Chevallier proved\\

\begin{theorem}[Chevallier \cite{chevallier91}] \label{bjq}
On $(E, \cdot, \textrm{or})$ the vector space of screws $\mathcal{S}$ can be given via Eqs.\ (\ref{jgg})-(\ref{jgy}) the structure of a rank 3 free $\mathbb{D}$-module endowed with an orientation and a scalar product (which we call the natural $\mathbb{D}$-module geometry structure of $\mathcal{S}$).
\end{theorem} $\empty$

Having given this introduction it now becomes easier to state the goal of this work in a precise way, and to also comment on the advantages of our approach. The objective of this work is to prove a converse of Chevallier's result. That is, we start from the following object\\
\begin{itemize}
\item[($\star\star$)] a triple $(M,\circ, OR)$ where $M$ is a rank 3 free $\mathbb{D}$-module endowed with an orientation `OR' and a scalar product `$\circ$' (which we call {\em $\mathbb{D}$-module geometric structure} or {\em $\mathbb{D}$-module geometry} for short),\\
\end{itemize}
and from it we recover  screw theory in both the approaches through equiprojective vector fields  and motors. The core of our argument consists in showing that from $M$ a natural structure of Euclidean space $E$ can be identified and that elements of $M$ (called with some abuse of terminology screws or motors) can be associated to screw vector fields on $E$.

We prove the next result\\

\begin{theorem} \label{bje}
Let $(M,\circ, \textrm{OR})$ be a $\mathbb{D}$-module geometric structure  (cf.\ ($\star\star$)), then there is a canonical 3-dimensional Euclidean space $E$ such that, denoting with $\mathcal{S}$ its space of screws, which we regard as endowed with its natural $\mathbb{D}$-module geometry structure (cf.\ Thm.\ \ref{bjq}),
there is a canonical ($\mathbb{D}$-linear) isomorphism $\beta\colon M \to \mathcal{S}$ preserving the scalar product and the orientation (and hence the cross product).
\end{theorem} $\empty$


Any  rank 3 free $\mathbb{D}$-module  endowed with a scalar product and an orientation is (non-canonically) isomorphic to $\mathbb{R}^3+\epsilon \mathbb{R}^3$, the isomorphism depending on the chosen positively oriented orthonormal basis. Thus any two such modules are (non-canonically) isomorphic. The main content of the previous theorem is  on the fact that $E$ and $\beta$ are constructed from $M$ in a canonical way. Thus the interesting aspects of the theorem are in the construction itself, for instance $E$ is defined as the set of $\langle \,,\rangle$-null real subspaces of $M$ of real dimension 3 transverse to $\epsilon M$, where  $\langle\cdot,\cdot\rangle=\mathfrak{Du}(\cdot \circ \cdot)$.

Subsequently, we show how the cross product of screws can be defined, and then all the relevant formulas  previously presented are derived going  in the reverse direction, namely moving from the module $M$ to the objects living in the Euclidean space $E$, rather than conversely. For instance, the dual angle is not constructed adjoining the distance to the ordinary angle in the combination $\theta+\epsilon d$, but rather emerges  from the definition of dual angle between two normalized screws. So equations (\ref{one})-(\ref{two}) are obtained naturally much  in analogy to what it is done in vector space geometry.

We are ready to recall the {\em transference principle} first formulated by Kotel'nikov \cite{kotelnikov95}. Broadly, it states that  most formulas of vector geometry find a formal analog in screw theory. For modern studies pointing out limitations, applications, and  more precise formalizations   of this principle the reader is referred to \cite{hsia81,selig86,martinez93,chevallier96,wittenburg16}.

The traditional approach to this principle is based on the rather surprising and puzzling  result that
formulas  (\ref{one})-(\ref{two}) hold true in screw theory, formulas that are then used to generalize several results in trigonometry \cite{dimentberg68}. In this approach it is difficult to grasp why (\ref{one})-(\ref{two})  had to hold in the first place, that seems to happen just by chance, and the very fact that defining the dual angle in that peculiar form brings such simplifications appears as quite surprising.

Instead, our approach makes it obvious that such formulas should hold, and in its new light the transference principle becomes no more surprising: it is sufficient to recognize that the starting structures, namely $(V,\cdot , \textrm{or})$ and  $(M,\circ, \textrm{OR})$  are pretty similar; that all the calculations for the latter find an analog in the former, and that all the calculations for the former that do not make use of coefficient invertibility, find an analog in the latter. This also clarifies that introducing the unit $\epsilon$, $\epsilon^2=0$, into the space of coefficients `removes the origin from the vector space' and brings us into the affine world while, and this is the important point, preserving linearity.

Thus by noting that the triple $(M,\circ, \textrm{OR})$ is perfectly analogous to the triple $(V,\cdot , \textrm{or})$ used in the definition of Euclidean space, and by observing that vector geometry is the geometry of the structure $(V,\cdot , \textrm{or})$, we arrive at the following broad  result which is our form of the transference principle


\begin{quote}
{\normalsize {\bf Transference principle}: All results of vector geometry that do not use in their proof, in any essential way, the invertibility of the $\mathbb{R}$-coefficients  of the vector space $V$, find a generalization to  results about screws (hence having an affine interpretation).}
\end{quote}

Actually, we recall that  in a $\mathbb{D}$-module the property of linear dependence does not imply the possibility of expressing one element as a linear combination of the others. Indeed, this result uses the invertibility of the coefficients and is rather of central importance. For this reason, the generalized results about screws alluded to above, will be about {\em proper} screws, namely screws with non-vanishing resultants as only those can belong to a basis of the module (as application of $\mathcal{E}$ shows, see also the below discussion).


As we shall see, another advantage of our approach is that it preserves the effectiveness of screw calculus while removing any reference to reduction points. In fact, it shows that any positively oriented orthonormal frame in Euclidean space $E$ is nothing but a positively oriented orthonormal basis on $M$. These basis elements on $E$ should not be seen as elements of $V\oplus \epsilon V$, so there is really no need to select a reduction point. Drawing a frame on $E$ means picking a basis on $M$ and, as a result, what matters are the dual matrices relating different frames, with really no need to convert the basis elements into dual vectors.

%
%

The structure of the paper is as follows. In Section \ref{sone} we recall some basic definitions about rings and modules, and prove that the notion of {\em rank} is well defined. In Section \ref{vkd} we show how the Euclidean space emerges out of $\mathbb{D}$-module geometry, in particular, Theorem \ref{bje} is proved. In Section \ref{stre} we obtain several results of screw theory by using the formalism of $\mathbb{D}$-module geometry. In Section \ref{squa} we obtain results that, a posteriori, can be interpreted as a manifestation of the transference principle. Again, they are obtained by using the formalism of  $\mathbb{D}$-module geometry.  We end the paper with some conclusions.

As for notation and terminologies, throughout the work we use, unless otherwise stated, the Einstein summation convention. Vectors
or module elements are denoted by lowercase letters in ordinary font.
In the context of a $\mathbb{D}$-module, by   {\em linear} we shall  mean $\mathbb{D}$-linear, while {\em $\mathbb{R}$-linearity} will always be written in this form.

\section{Basics on $\mathbb{D}$-modules} \label{sone}

Let us recall \cite[Sec.\ 2.1]{lang02}\\

\begin{definition}
A {\em ring} $A$ is a set, together with two laws of composition called multiplication
and addition respectively,
satisfying the following conditions:
\begin{enumerate}
\item  With respect to addition, $A$ is a commutative group.
\item The multiplication is associative, and has a unit element.
\item For all $x, y, z \in A$ we have (distributivity)
\[
(x + y)z = xz + yz.
\]
\end{enumerate}
\end{definition} $\empty$

It follows that for every $x\in A$, $x0=0$, and that $1\ne 0$ or we are in the trivial case $A=\{0\}$.
A ring $A$ is said to be {\em commutative} if $xy = yx$ for all $x, y \in A$. A ring $A$ such that $1 \ne 0$, and such that every non-zero element is
invertible is called a {\em division ring}. A commutative
division ring is called a {\em field}. Examples of fields are the real, $\mathbb{R}$, rational, $\mathbb{Q}$, and complex, $\mathbb{C}$, numbers.


As previously mentioned, the dual numbers
\[
\mathbb{D}=\mathbb{R}[\epsilon]/(\epsilon^2)=\{a+\epsilon b,\quad a,b\in \mathbb{R}\}=\mathbb{R}+\epsilon \mathbb{R}.
\]
form a commutative ring with unit and zero that clearly differ.\\

\begin{definition}
Suppose that $R$ is a commutative ring, and 1 is its multiplicative identity. An  $R$-module $M$ consists of an abelian group $(M, +)$ and an operation (empty symbol) $R \times M \to M$ such that for all $r, s \in R$ and $x, y \in M$, we have
\begin{align*}
 r (x+y)&= r x+r y,\\
 (r+s) x&=r x+s x, \\
 (rs) x&=r(s x), \\
 1 x&=x.
\end{align*}
\end{definition}

It is a generalization of the notion of vector space, wherein the field of scalars is replaced by a ring. It is called {\em free} if it admits a basis.
Not all modules admit a basis, for instance, $\epsilon \mathbb{R}$ is a $\mathbb{D}$-module but no element $z\in \epsilon \mathbb{R}$ gives a basis $\{z\}$ as $\lambda z=0$ admits the solution $\lambda=\epsilon\ne 0$.
It is {\em finitely generated} if there are some finite number of elements whose span is the whole $M$.

We provide the proof of the next result for completeness.\\

\begin{proposition}
For a finitely generated free $\mathbb{D}$-module $M$ every basis has the same number of elements. This number is called the rank of $M$.
\end{proposition}

\begin{proof}
Let $\{b_i\}$ and $\{b_i'\}$ be two bases for $M$, where the cardinality of the former is $n$ and that of the latter is $m$, and suppose, by contradiction, that $n< m<\infty$. There is a matrix $(A_{ij})\in \mathbb{D}^{m\times n}$ such that $b_i'=A_{ij} b_j$.
 Consider the equation $\lambda_r A_{ri} b_i=0$ which implies $\lambda_r A_{ri}=0$ for $i=1,\cdots, n$. The unknowns are $m$ in number.
 Suppose that there is some $i$ and some $r$ such that $A_{ri}$ is invertible, then $\lambda_r$ can be expressed in terms of the other unknowns. As long as there is a coefficient which is invertible we can continue the process. This shows that we shall be able to express $k\le n$ of the $m$ lambdas in terms of the other $m-k$ lambdas and that the $m-k>0$ remaining unknowns are going to satisfy $c\le n-k$ equations whose coefficients are purely dual numbers. If we set the remaining  $m-k$ unknowns equal to $\epsilon$ we satisfy all $n$ equations, which proves that $\lambda_r b_r'=\lambda_r A_{ri} b_i=0$ and hence the  $\{b_i'\}$ are not linearly independent. The contradiction proves $n \ge m$. Inverting the roles of the bases we get $m\ge n$ and hence $n=m$.
\end{proof}

\section{Screw theory from $\mathbb{D}$-module geometry} \label{vkd}


Let $M$ be a finitely generated free $\mathbb{D}$-module.
Clearly, it is isomorphic to the vector space $\mathbb{R}^{2n}$ where $n=\textrm{rank} M$, the isomorphism depending on the chosen basis.

The product by $\epsilon$ can be regarded as an endomorphism $\epsilon: M\to M$. Regarding $M$ as a real vector space we have $\textrm{dim}_{\mathbb{R}} M=\textrm{dim}_{\mathbb{R}} \textrm{ker} \epsilon + \textrm{dim}_{\mathbb{R}} \textrm{Im} \epsilon$. Since $\epsilon^2=0$ we have $\textrm{Im} \epsilon\subset \textrm{ker} \epsilon$. Note that the real dimension of $\textrm{Im} \epsilon$ must be at least $n$, indeed if $\{b_i\}$ is a basis of $M$, then $\{\epsilon b_i\}$ are $\mathbb{R}$-linearly independent by the $\mathbb{D}$-linear independence of $\{b_i\}$. This proves $\textrm{dim}_{\mathbb{R}} \textrm{Im} \epsilon\ge n$, and since  $\textrm{dim}_{\mathbb{R}} M=2n$ we conclude $\textrm{dim}_{\mathbb{R}} \textrm{Im} \epsilon=\textrm{dim}_{\mathbb{R}} \textrm{ker} \epsilon=n$ and hence $\textrm{Im} \epsilon= \textrm{ker} \epsilon$.

There is a natural vector subspace of $M$, namely $\epsilon M=\textrm{Im} \epsilon$, and a natural quotient space $V=M/\epsilon M$, both of real dimension $n$. As usual for linear maps, $\epsilon$ passes to the quotient to  an isomorphism between $M/\textrm{ker} \epsilon$ and $\textrm{Im} \epsilon$, namely between $V$ and $\epsilon M$. This map is denoted with the same letter, i.e.\ $\epsilon: V\to \epsilon M$, as this notation should not cause confusion.

\subsection{Real subspace of a basis and orientation }

Let $\{b_i\}$ and $\{b'_j\}$ be two bases (all our bases are ordered). We can write $b'_j=A_{jk} b_k$ for a suitable matrix of dual numbers. We say that $\{b'_i\}\sim \{b_i\}$ if $\det \mathfrak{Re}(A)>0$. Observe that if $b''_j=A_{jk} b'_k$ is a third basis, then $\mathfrak{Re}(AB)=\mathfrak{Re}(A) \mathfrak{Re}(B)$. From here it can be shown that $\sim$ is an equivalence relation and that there are just two classes.\\

\begin{definition}
We say that the module $M$ has been endowed with an  orientation (denoted OR as in ($\star\star$)) if one class of the equivalence relation has been chosen as `positive' (for $\textrm{rank} M=3$ the bases in such class might be referred as 'right-handed', 'left-handed' being the bases in the other class).
\end{definition} $\empty$

Let $\{b_i\}$ be a basis and let $R$ be the real $n$-dimensional subspace spanned by this basis.  No element of $x\in R$ can belong to $\epsilon M$, for otherwise $0=\epsilon x=(\epsilon x^i) b_i$ which contradicts the linear independence of $\{b_i\}$. We conclude that $R$ is  transverse\footnote{Two subspaces $L_1,L_2$ of a real vector space $M$ are {\em transverse} if $L_1+L_2=M$. If $\dim_{\mathbb{R}} L_1= \dim_{\mathbb{R}} L_2=(\dim_{\mathbb{R}} M)/2$, as in our case, then $L_1\cap L_2=\{0\}$ iff $L_1+L_2=M$ iff $L_1\oplus L_2=M$, the latter property holding if any vector  $m\in M$ can be written in a unique way as a sum: $m=l_1+l_2$, $l_1\in L_1$, $l_2\in L_2$.} to $\epsilon M$, so that $M=R\oplus \epsilon M$.

Let $R$ be a real $n$-dimensional subspace transverse to $\epsilon M$, $M=R\oplus \epsilon M$.
Since $R$ has dimension $n$, the quotient $\pi: M\to V$ is such that $\pi(R)=V$, thus $\pi$ provides an $\mathbb{R}$-linear isomorphism between $R$ and $V$. For every, $v\in V$ we denote with $v^R$ the unique vector in $R$ such that $\pi(v^R)=v$.
Every basis in $M$ projects to a basis of $V$ (because with $\lambda_i\in \mathbb{R}$, $0=\lambda_i \pi(b_i)=\pi(\lambda_i b_i)$ which implies $\lambda_ib_i\in \epsilon M$, hence $\epsilon \lambda_i b_i=0$ and finally $\lambda_i=0$), and every basis in $V$ is the projection of a basis in $M$ (because with $\lambda_i=a_i+\epsilon b_i\in \mathbb{D}$, $0=\lambda_i v_i^R$, implies $a_i v_i=0$, thus $a_i=0$ and hence $b_i v_i^R\in \epsilon M$, which projected gives $b_i v_i=0$ and hence $b_i=0$). Moreover, it is easy to verify that every two bases in the same class in $M$ project to bases in the same class in $V$, and similarly if two projected bases in $V$ are in the same class so are the original bases in $M$. This means that the orientation in $M$ induces an orientation in $V$ and conversely.

\subsection{Definition of scalar product}

When working on $M$ by scalar product we shall mean the following object.\\

\begin{definition} \label{bas}
We say that a  $\mathbb{D}$-bilinear map $\circ : M\times M\to \mathbb{D}$ is a {\em scalar product} if for every $x$, $\mathfrak{Re}(x\circ x) \ge 0$, with equality iff $x \in \epsilon M$ (in which case $x\circ x=0$, so that the elements of $\epsilon M$ can be called {\em null vectors}).
\end{definition} $\empty$

 In other words $\mathfrak{Re} (\cdot \circ \cdot)$ descends to a scalar product for the real vector space $V$. This scalar product on the quotient is denoted with a dot `$\cdot$'.

The square of $x\in M$, is denoted $x^2:=x\circ x$.
The modulus of  $x\in M\backslash \epsilon M$ is defined as $\vert x \vert :=\sqrt{x^2}$, and as 0 if $x\in \epsilon M$. Notice that the modulus is a dual number and $x^2=\vert x \vert^2$.

Let $\{b_i\}$ be a basis of $M$, the Gram-Schmidt orthogonalization proceeds formally as usual
\begin{align*}
c_1&=b_1,\\
c_2&=b_2-\frac{b_2\circ c_1}{c_1\circ c_1}\,c_1,\\
c_3&=b_3-\frac{b_3\circ c_1}{c_1\circ c_1}\, c_1-\frac{b_3\circ c_2}{c_2\circ c_2}\, c_2
\end{align*}
and so on. Then one defines $m_i=c_i/\vert c_i\vert$ (no sum over $i$).

In the process we cannot get a first $c_i\in \epsilon M$, for multiplying the corresponding equation by $\epsilon$ we would get an equation contradicting the linear independence of $\{b_i\}$. Thus each $c_i$ is such that $c_i\circ c_i$ is invertible, so that the process indeed makes sense.
We conclude that $M$ admits orthonormal bases $\{m_i\}$ (which, given an orientation, can be found positive).

If $z\in M$ does not belong to $\epsilon M$ then it belongs to a basis of $M$ and, if it is normalized it belongs to an orthonormal basis. Indeed, let $\{b_i\}$ be a basis then $z=z^i b_i$. Not all coefficients $z^i$ can be pure dual otherwise $\epsilon z=0$ which would imply $z\in \epsilon M$, a contradiction. Thus some coefficient is invertible and hence we can express some basis elements in terms of $z$ and the other basis elements. This shows that $z$ belongs to a basis. The  Gram-Schmidt orthogonalization preserves the first basis element, if it is normalized, hence the second claim.

\subsection{Definition of cross product and mixed product}

From now on assume that $M$ is a free $\mathbb{D}$-module of  rank 3, endowed with a scalar product $\circ$, and an orientation $\textrm{OR}$. We denote this triple with $(M, \circ,$ $ \textrm{OR})$. We can define on $M$ a cross product  in the usual way
\[
x\times y=x^i y^j \epsilon_{ijk} m_k ,
\]
where $x=x^i m_i$, $y=y^i m_i$ and $\{m_i\}$ is a positive orthonormal basis. Indeed, the independence from the choice of basis follows from the fact that  the change between positive orthonormal bases has the following form
\begin{equation}
m'_i=U_{ik} m_k, \qquad U_{ik}= O_{ij}\,(\delta_{jk}+\epsilon A_{jk})
\end{equation}
where $O_{ij}$ is real special orthogonal and $A_{jk}$ is real antisymmetric. The determinant of this product matrix is just 1.

The mixed product can be expressed in a basis as usual $x\times y \circ z=x^i y^i z^k \epsilon_{ijk}$ which shows that it is antisymmetric and hence invariant under cyclic permutations. The standard identities for the Levi-Civita symbol imply the usual formulas for the double cross product
\[
x \times (y\times z)=(x \circ z)\, y - (z \circ y)\, z.
\]
We have also the formulas
\begin{align}
(x \times y) \circ (u \times  v) &= (x \circ  u)\,(y \circ v) - (x \circ v)\,(y \circ u), \label{ciq} \\
(x \times y) \times (u \times  v) &=  [x,u,v]\, y- [y,u,v]\, x=[x,y,v]\,u-[x,y,u] \,v,
\end{align}
with $[x,y,z]:=x\circ y \times  z$,
and the Jacobi identity
\begin{equation} \label{hhf}
x \times (y\times z)+z \times (x\times y)+y \times (z\times x)=0.
\end{equation}

\subsection{Definition of screw scalar product and $\langle,\rangle$-null subspaces}
Let us denote $\langle x, y \rangle:=\mathfrak{Du}(x \circ y)$, and let us call it {\em screw scalar product}. A $\langle,\rangle$-{\em null subspace} $S\subset M$, is a real vector subspace such that for any $x,y\in S$, $\langle x, y\rangle=0$.

The real span $R$ of a basis $\{m_i\}$ of $M$ is  $\langle,\rangle$-null. Indeed,  for any two elements $x,y\in R$
\[
x\circ y= x^i y^j m_i \circ m_j= x^i y^j \delta_{ij}\in \mathbb{R} \ \Rightarrow \  \langle x, y \rangle=0
\]
as $x^i,y^i\in\mathbb{R}$. Clearly, $\epsilon M$ is $\langle,\rangle$-null as it is $\circ$-null.

We also know that $R$ is transverse to $\epsilon M$ and $M=R\oplus \epsilon M$.
Any $\circ$-orthonormal basis $\{e_i\}$ in $R$ projects to a $\cdot\,$-orthonormal basis on $(V,\cdot)$. Conversely, given a $\langle,\rangle$-null subspace of real dimension $n$ transverse to $\epsilon M$, every orthonormal basis in $V$ lifts to a unique $\circ$-orthonormal basis on $R$.\\

\begin{proposition}
The symmetric bilinear form $\langle,\rangle$ on $M$, where $M$ is regarded as a 6-dimensional real vector space, has signature $(+,+,+,-,-,-)$, and the $\langle,\rangle$-null subspaces have at most real dimension~3.
\end{proposition}

\begin{proof}
Let $\{m_i\}$ be a positive orthonormal basis, then $\{m_i, \epsilon m_i\}$ provide a basis for $M$ regarded as a 6-dimensional real vector space. Since
every element $x\in M$ reads $x=(a^i+\epsilon b^i)m_i$, we have $\langle x,x \rangle= 2a^ib^j\delta_{ij}=\frac{1}{2} \sum_{i=1}^3 [(a^i+b^i)^2- (a^i-b^i)^2$, which proves that the signature is $(+,+,+,-,-,-)$. The  $\langle,\rangle$-null subspaces are those real subspaces over which the bilinear form induced from  $\langle,\rangle$ vanishes. By the Cauchy interlace theorem \cite{parlett98} they have at most dimension 3.
\end{proof}

\subsection{Construction of the (affine) Euclidean space}


Let $E$ be the set of  $\langle,\rangle$-null subspaces of real dimension $3$ transverse to $\epsilon M$.
We want to show that there is a map (subtraction) $-:E\times E \to V$, $(A,B)\mapsto B-A$, which converts $E$ into an affine space associated to the real vector space $V$. That is, recalling Weyl's axioms of an affine space \cite[Sec.\ I.2]{nomizu94}, our goal is to define the subtraction and prove the properties:
\begin{itemize}
\item[(a)] for every $A\in E$ and every $x\in V$ there is one  and only one $B\in E$ such that $B-A=x$,
\item[(b)] for every $A,B,C\in E$, we have $(B-A)+(C-B)=C-A$.
\end{itemize}
Since $V$ has both a scalar product and an orientation (induced from that of $M$) we would have that $E$ is a Euclidean 3-dimensional space.

First, we need to define the map $-:E\times E \to V$.
Let $A, B\in E$ and let $\{e_i\}$ be a positive orthonormal basis of $V$. Let $\{e^A_i\}$ and $\{e^B_i\}$  be the lifts to $A$ and $B$, respectively. Since the former is a basis for $M$ we have $e^B_i= (C_{ij}+\epsilon D_{ij})e^A_j$ where $C$ and $D$ are real matrices. But since for each $i$, the vectors $e^B_i$ and $e^A_i$ project on the same element of $V$ they differ just by an element of $\epsilon M$, thus $C$ is the identity. The $\circ$-orthonormality of $\{e^A_i\}$  and $\{e^B_i\}$ implies that $D$ is antisymmetric. We can write $D_{ij}=\epsilon_{ijk} d^k$. The vector $d:=d^k e_k\in V$ does not depend on the initial choice of basis of $V$, since any two such bases are related by a real special orthogonal matrix. We set $B-A:=d$.

Let us prove (b): if $C\in E$ is a third subspace  and $\{e^C_i\}$ the corresponding basis,
\begin{align*}
e^C_i&= (\delta_{ij}+\epsilon F_{ij})\, e^B_j=(\delta_{ij}+\epsilon F_{ij})(\delta_{jk}+\epsilon D_{jk})\, e^A_k\\
&=[\delta_{ik}+\epsilon (D_{ik}+F_{ik})]\, e^A_k=[\delta_{ik}+\epsilon \,
\epsilon_{ikt} (d^t+f^t)]\, e^A_k,
\end{align*}
 which implies $(B-A)+(C-B)=C-A$.

 Let us prove (a): just expand $x$ with respect to the positive orthonormal basis of $V$, $x=x^k e_k$, and define $B$ to be the real subspace spanned by the basis $e^B_i= (\delta_{ij}+\epsilon \epsilon_{ijk} x^k)e^A_j$, then $B-A=x$ (observe that $B$ is $\langle,\rangle$-null as $\mathfrak{Du}(e^B_i\circ e^B_j)=\mathfrak{Du}(\delta_{ij})=0$). We just proved that $B$ exists. In order to show that it is unique suppose that $B'-A=x$, then $e^{B'}_i= (\delta_{ij}+\epsilon \epsilon_{ijk} x^k)e^A_j=e^{B}_i$. As the spans of these bases coincide, $B'=B$, which shows that $B$ is unique.

\subsection{Displacement from bases}

The previous paragraph gives a pratical method to calculate the displacement between points using the module $M$. Each point of $E$ can be represented by an orthonormal positive basis of $M$. So let $\{m_i\}$ and $\{m_i'\}$ be two orthonormal positive bases, let $O,O'\in E$ be the corresponding points, and rotate the latter basis by a real special orthogonal matrix if necessary so that they project to the same basis $\{e_i\}$ of $V$, then we have the formula
\begin{equation} \label{vpd}
O'-O=\frac{1}{2} \epsilon^{-1} m_i \times m_i'.
\end{equation}
Indeed, we have shown above that $m_i'=(\delta_{ij}+\epsilon \epsilon_{ijk}(O'-O)^k)m_j$  thus
\[
 m_i \times m_i'= (\delta_{ij}+\epsilon \,\epsilon_{ijk}(O'-O)^k) \, \epsilon_{ijr} m_r= \epsilon 2 (O'-O)^k  m_k.
\]
  Notice that $1/\epsilon:=\epsilon^{-1} : \epsilon M \to V$ is well defined as the inverse of $\epsilon: V\to \epsilon M$, and the product $m_i \times m_i'$ is indeed in $\epsilon M$ as its resultant vanishes. The right-hand side will be further elaborated below, see Eq.\ (\ref{ipd}).

\subsection{Induced screw vector field (and end of proof of Theorem \ref{bje})} \label{jxe}

We want to show that  every $z\in M$ induces, in a canonical way, a screw vector field $\mathcal{s}: E\to V$ and finally construct a ($\mathbb{D}$-linear) bijection $\beta:M \to \mathcal{S}$ which preserves the cross product, the scalar product and the orientation.

Let us denote the quotient projection of $z$ (which ultimately will be the resultant of the screw field associated to $z$) with $s=\pi(z)\in V$ (see Figure \ref{figureone}).

First observe that for every $A\in E$, we have a splitting $M=A\oplus\epsilon M$ and $\epsilon$ establishes an isomorphism between $V$ and $\epsilon M$, thus we can write
\begin{equation} \label{bka}
z=a+ \epsilon \,\mathcal{s}(A)=s^A+\epsilon \, \mathcal{s}(A)
\end{equation}
 with $a\in A$ and $\mathcal{s}(A)\in V$.

 \begin{figure}[ht]
\centering
\includegraphics[width=6cm]{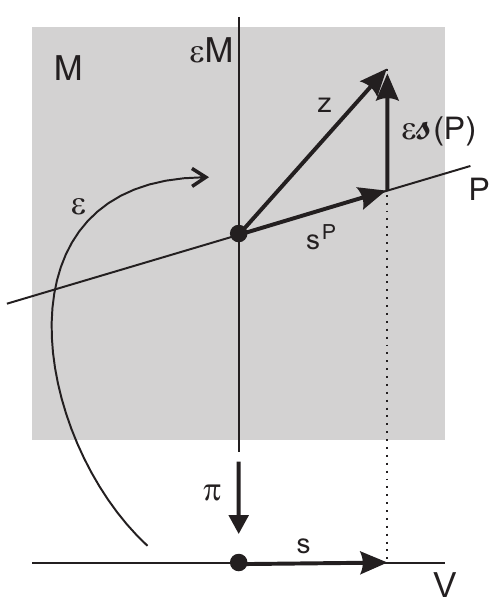}
\caption{The gray region is the module $M$, the black dot at its center is its origin. At the bottom we depicted  $V=M/\epsilon M$, the quotient 3-dimensional vector space. The map $\epsilon\colon M\to M$, $m \mapsto \epsilon m$, has domain $M$ but it passes to the quotient as a map $\epsilon\colon V\to M$,  where it is given the same letter (cf.\ first paragraphs of Sec.\ \ref{vkd}). It is an isomorphism between $V$ and its image $\epsilon M$. Here $P$ is  a  $\langle,\rangle$-null subspace of real dimension $3$ transverse to $\epsilon M$, hence $M=P\oplus \epsilon M$ (the set of all such  $P$ is, by definition, the set $E$ which we prove is actually an affine space associated to $(V,\cdot, or)$ hence a 3-dimensional Euclidean space). Each $P$ induces a splitting that selects a different `vertical' part which determines $\mathcal{s}(P)$ via the isomorphism $\epsilon M \leftrightarrow V$, see Eq.\ (\ref{bka}). Finally, $P\mapsto \mathcal{s}(P)$ is the screw field induced on Euclidean space by $z\in M$. Its  resultant $s$ is the quotient projection of $z$ which does not depend on $P$.}
\label{figureone}
\end{figure}

Suppose that we choose a different subspace $B\in E$, then $z=b+\epsilon \mathcal{s}(B)$. Let $\{e_i\}$ be a positive orthonormal basis in $V$ and $\{e_i^A\}$ and $\{e^B_i\}$ bases of $A$ and $B$ that project on $\{e_i\}$ (necessarily positive and $\circ$-orthonormal). We have $z=a^i e_i^A+\epsilon  \mathcal{s}(A)^i  e_i^A$ and $z=b^i e_i^B+\epsilon  \mathcal{s}(B)^i  e_i^B$. Since they have the same projection $s\in V$, $s=s^i e_i$ we have $b^i=a^i=s^i$. Recall that
$e_i^B= e_i^A+\epsilon \epsilon_{ijk} d^k e_j^A$, thus
\[
z=s^i e_i^A+\epsilon[ s^i\epsilon_{ijk} d^k +\mathcal{s}(B)^j  ]\,e_j^A
\]
which implies $ s^i\epsilon_{ijk} d^k +\mathcal{s}(B)^j = \mathcal{s}(A)^i$, or equivalently
\begin{equation} \label{bjm}
\mathcal{s}(B)-\mathcal{s}(A)= s \times (B-A).
\end{equation}
Conversely, given a screw field $\mathcal{s}: E\to V$ we can associate to it a unique element $z\in M$ via Eq.\ (\ref{bka}). This map is really independent of the chosen element $A\in E$ as the last calculations followed in the backward direction show.

The just constructed bijection $\beta:M \to \mathcal{S}$, $z\mapsto \mathcal{s}$, is actually $\mathbb{R}$-linear as Eq.\ (\ref{bka}) shows.
In order to prove that it is $\mathbb{D}$-linear we need only to show that $\beta(\epsilon z)=\mathcal{E}(\beta(z))$. Now, since $\epsilon z=\epsilon s\in \epsilon M$, we have for every $A\in E$, $\beta(\epsilon z)(A)=s$, which means indeed $\beta(\epsilon z)=\mathcal{E}(\beta(z))$.

Let $z_1,z_2\in M$, then
\begin{align}
\begin{split} \label{bqd}
z_1\circ z_2&=\left(s_1^i e_i^A+\epsilon  \mathcal{s}_1(A)^i e_i^A\right) \circ \left(s_2^j e_j^A+\epsilon  \mathcal{s}_2(A)^j e_j^A \right)\\
&=s_1^i s_2^j
 \delta_{ij}
+\epsilon \left[s_1^i  \mathcal{s}_2(A)_i +\mathcal{s}_1(A)_i s_2^i\right] \\
&=s_1\cdot s_2+\epsilon \left[s_1\cdot \mathcal{s}_2(A)+\mathcal{s}_1(A)\cdot s_2\right].
\end{split}
 \end{align}
 This equation implies that over a $\langle,\rangle$-null subspace the $\circ$-scalar product coincides with the $\cdot $\,-scalar product of the projections (resultants).
Using $e_i^A\times e_j^A=\epsilon_{ijk} e_k^A$
\begin{align}
\begin{split} \label{dir}
z_1\times z_2&=\left(s_1^i e_i^A+\epsilon  \mathcal{s}_1(A)^i e_i^A\right) \times \left(s_2^j e_j^A+\epsilon  \mathcal{s}_2(A)^j e_j^A \right)\\
&=s_1^i s_2^j
 \epsilon_{ijk} e_k^A
+\epsilon \left[s_1^i  \mathcal{s}_2(A)^j +\mathcal{s}_1(A)^i s_2^j\right] \epsilon_{ijk} e_k^A \\
&=(s_1\times s_2)^A+\epsilon \left[s_1\times \mathcal{s}_2(A)+\mathcal{s}_1(A)\times s_2\right]
\end{split}
 \end{align}
which shows that the resultant of $z_1\times z_2$ is $s_1\times s_2$. This result also shows that $\beta(z_1\times z_2)=-[\beta(z_1),\beta(z_2)]_L$ where $[,]_L$ is the Lie bracket of two fields, in other words, $\beta$ is a Lie algebra isomorphism between $(M,-\times)$ and $(\mathcal{S},[,]_L)$.
Moreover,
\begin{align*}
z_1\times z_2 \circ z_3=s_1\times s_2 \cdot s_3+&\epsilon [\mathcal{s}_1(A)\cdot s_2\times s_3 \\
& \quad +\mathcal{s}_3(A)\cdot s_1\times s_2+\mathcal{s}_2(A)\cdot s_3\times s_1].
 \end{align*}
The fact that the combination $s_1\cdot \mathcal{s}_2(A)+\mathcal{s}_1(A)\cdot s_2$ does not depend on $A$ follows immediately from the fact that it coincides with $\mathfrak{Du}(z_1\circ z_2)$ which does not depend on it. Moreover, the fact that $s_1\times \mathcal{s}_2(A)+\mathcal{s}_1(A)\times s_2$ transforms as a screw field under change of $A$ follows from the fact that $z_1\times z_2\in M$. Observe that if $z_1$ and $z_2$ have proportional resultants and $z_1\circ z_1=z_2\circ z_2=1$ then $s_1\cdot \mathcal{s}_2(A)\propto s_2\cdot \mathcal{s}_2(A)=0$ and similarly with 1 and 2 exchanged, thus $\mathfrak{Du}(z_1\circ z_2)=0$ (this fact will be important for the introduction of the dual angle).
If $z\in M$ is normalized, $z\circ z=1$, we might denote its screw vector field with the letter $\ell$ (because its geometric data on $E$ corresponds to an oriented line, see below).

\subsection{Equivalence between bases in $M$ and orthonormal \\ frames in $E$} \label{bjg}

Let $\{m_i\}$ be an orthonormal basis in $M$, and let $(O,\{e_i\})$ be the pair given by (a) the 3-dimensional real subspace spanned by $m_1,m_2,m_3$ (which is actually $\langle,\rangle$-null and transverse to $\epsilon M$), and (b) the projected basis on $V$. This obtained data can be used to recover $\{m_i\}$ and hence is equivalent to it, indeed it is sufficient to lift the basis $\{e_i\}$ as $m_i=e_i^O$. Pictorially, we can thus represent $\{m_i\}$ as a usual frame in the Euclidean space $E$.

We mentioned that formula (\ref{vpd})  giving the displacement between the origins of two frames $(O,e_i)$, $(O', e_i')$, can be written
\begin{equation} \label{ipd}
O'-O=\frac{1}{2} \left[e_i\times \ell_i'(A)+\ell_i(A)\times e_i'\right]=\frac{1}{2} \left[e_i\times \ell_i'(O)\right]
\end{equation}
where $A$ is arbitrary and so can be chosen equivalent to $O$. The last expression might be interpreted as the displacement expression in terms of motor reduction, see next section.

\subsection{Motor reduction and $\mathbb{D}$-linear isomorphism}

Now observe that for every $A\in E$, the map $\phi_A: M\to V\otimes \epsilon V$, $z\to \pi(z)+\epsilon \mathcal{s}(A)$ is a $\mathbb{D}$-linear isomorphism which preserves the scalar product and cross product. This is the motor reduction at $A$. In most applications the motor reduction will not be needed. Once we know that orthonormal frames  in $E$ are given by orthonormal bases in $M$, the only thing we need to know is that they are related by dual matrices, and these do not depend on any motor reduction.


\subsection{Definition of screw axis and characterization in terms of screw field}

We have shown above that every normalized $z\in M$ belongs to some positive orthonormal basis  and hence is contained in some $\langle,\rangle$-null subspace of real dimension $3$ transverse to $\epsilon M$, namely in some point $O\in E$. Let $\{e_i\}$ be a basis of  $V$ and let $z=z^i e_i^O$, $z^i\in \mathbb{R}$. But $e_i^O=e_i^P+\epsilon A_{ij} e^P_j$ with $e_i^P$ basis with the same projection  and $A_{ij}$  is a real antisymmetric matrix. Thus $z=z^i e_i^P+ \epsilon z^iA_{ij} e^P_j$. This shows that $z$ belongs to $P$ iff $z^iA_{ij}=0$ which is true iff $A_{ij}=\lambda \epsilon_{ijk} z^k$ for some $\lambda\in \mathbb{R}$, which means $P-O=\lambda z^i e_i$. Thus  $z$ belongs precisely to the points of a whole  line $\{O+\lambda z^i e_i, \lambda \in \mathbb{R}\}$. We conclude that the normalized elements of $M$ represent oriented lines on $E$. Every element of $z\in M\backslash \epsilon M$ can be normalized so it can be written in the form $z=(a+b \epsilon) u=a u +\epsilon bu$ with $a>0$ and $u\circ u=1$.
Indeed, $a+b\epsilon =\sqrt{z \circ z}$ and $u=(a+\epsilon b)^{-1} z$.
By using  the pitch $p=b/a$, it can be written in the form $z= a e^{\epsilon p} u$.
The {\em screw axis of $z$} is, by definition, the line of $u$. Observe that if we take a point $A\in E$, such that $u\in A$  (remember that $u$ belongs to a $\circ$-orthonormal basis and hence to the real 3-dimensional subspace $A$ spanned by it) then the screw can be written $z=s^A+\epsilon  \mathcal{s}(A)$, as $z=au+\epsilon b u=au+\epsilon b \pi(u)$ and $au$ belongs to $A$ (since $a$ is real) we have $s^A=au$ (thus $s=a\pi(u)$) and $\mathcal{s}(A)=b\pi(u)$. This shows that for $A$ in the screw axis the screw field $\mathcal{s}(A)$ is aligned with the resultant and independent of the point in the screw axis. Moreover, Eq.\ (\ref{bjm}) shows that there are no other points where the screw field is aligned with the resultant, thus in terms of the screw field the axis is characterized as the locus where the screw field is proportional to the resultant.

Pictorially, we can represent  an element of $M\backslash \epsilon M$ as an oriented line on $E$ aligned with the resultant $a\pi(u)$ and a second vector $\mathcal{s}_\parallel:=b\pi(u)=\mathcal{s}(A)$ called {\em vector invariant}. If the pitch  vanishes (equivalently $z\circ z \in \mathbb{R}$) we also speak of {\em applied, line or sliding vector} \cite{dimentberg68,huang13}.

\subsection{Properties of sliding vectors}

\begin{proposition}
The points of $E$ are precisely the real subspaces of $M$ consisting of just sliding vectors (save for the zero vector) of maximal real dimension (necessarily 3).
\end{proposition}

\begin{proof}
Note that every point $A$ of $E$, being a $\langle,\rangle$-null subspace, is a real subspace of real dimension 3 consisting of just sliding vectors (save for the zero vector). Indeed, $A\cap \epsilon M=0$ which means that any non-vanishing element $z\in A$ has resultant, and zero pitch as $\mathfrak{Du}(z\circ z)=\langle z,z\rangle=0$.

Conversely, let  $S$ be a real subspace  consisting of just sliding vectors (save for the zero vector). By the polarization formula it is $\langle,\rangle$-null. It cannot have elements in common with $\epsilon M$ but the zero vectors, for such elements would have vanishing resultant contradicting the definition of sliding vector. Thus $S\cap \epsilon M =\{0\}$  and hence $S$ can have at most real dimension 3.
\end{proof}

Consider $N$ sliding vectors and assume that their screw axes intersect at a point $O$ (they are concurrent). This means that $z_i\in O$ for each $i$ and since they have zero pitch $\mathcal{s}_i(O)=0$, hence $z_i=v_i^O$,  for some $v_i\in V$. This implies $\sum_i z_i=(\sum_i v_i)^O$ namely the sum of sliding vectors whose axes intersect is again a sliding vector whose resultant is the sum of the resultants.

It can also be observed that the cross product of two sliding vectors whose axes intersect is again a sliding vector whose axis passes through the intersection point. Finally, the scalar product of two sliding vectors whose axes intersect coincides with the scalar product of the resultants.
 Many  calculations that we shall meet in geometrical plane problems will be simplified by these observations since in a plane any two non-parallel lines intersect.

\section{Further results obtained via $\mathbb{D}$-module geometry} \label{stre}

In this section we continue obtaining results on the correspondence between concepts from  $\mathbb{D}$-module geometry, Euclidean geometry and screw theory.

\subsection{Lie algebra considerations} \label{mff}

Let $G$ be the Lie group of isometries of $(M,\circ,OR)$, that is, linear maps $U: M\to M$ such that for every $v,w\in M$, $U(v)\circ U(w)=v\circ w$.
We shall be interested in the connected component of the identity, so that the determinant of any of its elements is 1. We denote $SO(3,\mathbb{D})$ such connected component.

Let $U(t)$ be a 1-parameter family passing through the identity $U(0)=I$. By differentiating, evaluating at $t=0$, and denoting  $B=U'$ we get
\begin{equation} \label{aub}
B(v) \circ w +v \circ B(w)=0,
\end{equation}
namely the Lie algebra,  $so(3,\mathbb{D})$ consists of $\circ$-antisymmetric matrices.\\

\begin{proposition}
The Lie algebra $so(3,\mathbb{D})$ is the family of $\circ$-antisymmetric operators. Every such operator $B: M\to M$ can be uniquely written $B=b\times$  with $b\in M$, and every operator of the form $b\times$ is $\circ$-antisymmetric .
\end{proposition}

\begin{proof}
For every $b\in M$, $B:=b\times$ satisfies the equation (\ref{aub}) (due to the properties of the mixed product) and generates a 1-parameter subgroup that belongs to $SO(3,\mathbb{D})$.

Conversely, if $B$ is $\circ$-antisymmetric, chosen a basis $\{m_i\}$, and defining
\[
b=\frac{1}{2} m_i \times B(m_i)
\]
we get
\begin{align*}
{ b}\times x&=\frac{1}{2} [m_i \times B(m_i)]\times x=\frac{1}{2}\left\{(x \circ m_i)  B(m_i)-(x \circ m_i))\, m_i  \right\}\\
&= \frac{1}{2}\{ B( (x \circ m_i)\, m_i)+(B(x) \circ m_i) \,m_i  \}=\frac{1}{2} [B(x)+B(x) ]=B(x)
\end{align*}
where we used twice the $\circ$-antisymmetry of $B$, and, twice, the fact that every element $z\in M$ can be written $z=(z\circ m_i) m_i$. If $b'\times =b \times$, then for every $z\in M$, $(b'-b)\times z=0$, which implies $b'=b$.
\end{proof}


\begin{proposition}
The Lie algebra $so(3,\mathbb{D})$ of $(M,\circ,OR)$ is canonically isomorphic with $se(3)$ and with $(M,-\times)$.
\end{proposition}

\begin{proof}
The Lie algebra $se(3)$ consists of the infinitesimal isometries of Euclidean space and can be identified $(\mathcal{S},[,]_L)$, that is, with the Lie algebra of screw fields. We found in Sec.\ \ref{jxe} that $\beta:M\to \mathcal{S}$ establishes a Lie algebra isomorphism between $(M,-\times)$ and $(\mathcal{S},[,]_L)$. Now if $a\times, b\times \in so(3,\mathbb{D})$ we have $[a\times, b\times ]x=a\times(b\times x))-b\times (a\times x)=(a\times b)\times x$ by Jacobi formula, which means that $\gamma: M\to so(3,\mathbb{D})$, $x\mapsto -x\times$ establishes a Lie algebra isomorphism between $(M,-\times)$ and $so(3,\mathbb{D})$.
\end{proof}

A related result was obtained in \cite{selig86} \cite[Sec.\ 7.6]{selig05}.

\subsection{Characterization of $\mathbb{D}$-linear independence for screws}

We recall that, with reference to the module $M$, by {\em linearity} we understand $\mathbb{D}$-linearity, sometimes writing in the last form just to stress it. Instead, when referring to the vector space $V$, by {\em linearity} we understand $\mathbb{R}$-linearity, of course.\\

\begin{proposition}
The elements of a finite subsets $\{z_i\}$, $z_i \in  M\backslash \epsilon M$, are $\mathbb{D}$-linearly independent iff so are their projections (resultants).
\end{proposition}

\begin{proof}
Suppose that $\{z_i\}$ are linearly independent, then with $a_i\in \mathbb{R}$, $0=a_i \pi (z_i)=\pi(a_i z_i)$ implies $a_iz_i\in \epsilon M$, thus $(\epsilon a_i)z_i=0$ which implies $a_i=0$. Conversely, suppose that $\{\pi(z_i)\}$ are linearly independent, and let $\lambda_i=a_i+\epsilon b_i$. If $\lambda_iz_i=0$ we have projecting, $a_i\pi(z_i)=0$, which implies $a_i=0$ hence $\epsilon b_i z_i=0$, which implies $b_iz_i\in \epsilon M$, thus projecting this expression, $b_i\pi(z_i)=0$, which implies $b_i=0$.
\end{proof}

\subsection{Characterization of proportionality of two screws}

Two elements $z_1,z_2$ can be $\mathbb{D}$-linearly dependent without one being proportional to the other. Remember that $z_1,z_2\in M\backslash \epsilon M$ can be normalized without changing their axes. Assume they are normalized. If they are proportional, i.e.\ $z_1=\pm z_2$, then they have the same unoriented screw axis. Conversely, if they have the same unoriented screw axis $e_1=\pm e_2$, and as the axes have a point $O$ in common, from the formula $z_i=e_i^O$ we get $z_1=\pm z_2$.

We conclude that $z_1,z_2\in M\backslash \epsilon M$ are $\mathbb{D}$-linearly dependent but not proportional iff their unoriented axes have the same direction but do not coincide.

\subsection{Characterization of proportionality of three  screws}

As for $z_1,z_2,z_3\in M\backslash \epsilon M$, it can be that while $z_1,z_2$ are $\mathbb{D}$-linearly independent, $z_1,z_2, z_3$ are linearly dependent (and so the resultants $r_1,r_2,r_3$ are linearly dependent) but $z_3$   is not proportional to a linear combination of $z_1,z_2$.\\

\begin{proposition} \label{vbb}
Let  $z_1,z_2,z_3\in M\backslash \epsilon M$ and assume that $z_1,z_2$ are $\mathbb{D}$-linearly independent. We have $z_3=a z_1+bz_2$ for some $a,b\in \mathbb{D}$, iff the three axes intersect orthogonally the same line.
\end{proposition}

\begin{proof}
Observe that if $z_3=a z_1+bz_2$ with $a,b\in \mathbb{D}$, then applying $z_4 \circ$ with $z_4$ the oriented line orthogonal to the screw axes of $z_1$ and $z_2$ we get $z_4 \circ z_3=0$ which means that the axes of $z_1,z_2,z_3$ are orthogonal to the same axis.

For the converse assume that the three axes intersect orthogonally the same line then $r_1,r_2,r_3$ are linearly dependent, hence $r_3=\alpha r_1+\beta r_2$ with $\alpha,\beta\in \mathbb{R}$. Let $a,b\in \mathbb{D}$ be such that $\mathfrak{Re}(a)=\alpha$, $\mathfrak{Re}(b)=\beta$, then their dual parts can be chosen so as to generate the desired moment of $z_3$ at the intersection between the axis of $z_3$ and the common normal.
\end{proof}

\begin{proposition} \label{vbf}
Let  $z_1,z_2,z_3\in M\backslash \epsilon M$ and assume $z_1+z_2+z_3=0$, then the axes   intersect orthogonally a common line. Moreover, if the axes are parallel then they are coplanar. Finally, if all $z_i$ have  zero pitch then the axes are coplanar and concurrent.
\end{proposition}

\begin{proof}
If the axes are not parallel then  they  intersect orthogonally a common line as it follows from Prop.\ \ref{vbb}.

If they are parallel, we can write $-z_3=z_1+z_2$ then the moment generated by $z_1$ and $z_2$ on the axis of $z_3$ must be aligned with the same axis. If the axes are not complanar  this is impossible as the applied resultants $r_1$ and $r_2$ would generate two moments orthogonal to the axis of $z_3$ that have different direction and hence that do not cancel. Thus they are coplanar and hence intersect orthogonally a common line.

If two axes are parallel all of them are and so they are coplanar. If no pair of axes are parallel, the three axis intersect a common line. Let $P_3$ be the intersection between $z_3$ and the common line. As $z_3$ has zero pitch its screw field vanishes on $P_3$, but $-z_3=z_1+z_2$ and the total screw field generated by the other two elements cannot vanish unless they are coplanar and concurrent on $P_3$.
\end{proof}

\subsection{Characterization of intersection of axes of unit screws via $\mathfrak{Du}(x\circ y)=0$}

\begin{proposition}
Let $x,y\in M\backslash \epsilon M$, be normalized, i.e.\ $x \circ  x= y \circ  y=1$, and let the resultants be $\mathbb{D}$-linearly independent (so that $x$ and $y$ are linearly independent). We have  $\mathfrak{Du}(x\circ y)=0$ iff their screw axes intersect.\\
\end{proposition}

As a consequence, $x\circ y=0$ iff their axes intersect orthogonally. We say that the axes (or $x$ and $y$) are {\em  normal}.

\begin{proof}
Assume $\mathfrak{Du}(x\circ y)=0$.
We can find a third linearly independent element (just choose its resultant so that it is linearly independent with the other two), define $u= y- ( x \circ  y) x$ and proceed with the Gram-Schmidt orthogonalization. So there is an orthonormal basis $(x,u,v)$ whose real span is a point $O$ of $E$. But due to $\mathfrak{Du}(x\circ y)=0$, $y$ is in the real span of $(x,u)$ and hence in $O$. This shows that if two normalized screws (lines) are such that $\mathfrak{Du}(x\circ y)=0$ then their screw axes intersect (hence if $x\circ y=0$ then they intersect orthogonally).

Conversely, if the axes intersect we can find $O\in E$ in both, namely $x,y\in O$ where $O$ is a  $\langle,\rangle$-null subspace of real dimension $3$ transverse to $\epsilon M$. This implies  $x\circ y\in \mathbb{R}$ as $O$ is $\langle,\rangle$-null, and hence $\mathfrak{Du}(x\circ y)=0$.
\end{proof}


\subsection{Definition of dual angle of screws and scalar product}

Let $x,y\in M\backslash \epsilon M$, $x\circ x=y \circ y=1$.
For every $\lambda \in \mathbb{R}$ we have  $\mathfrak{Re}((x+\lambda y)\circ (x+\lambda y)) \ge 0$ which leads to the Cauchy-Schwarz inequality
\[
\vert\mathfrak{Re}(x\circ y)\vert \le  1,
\]
 where equality holds iff there is some $\bar \lambda$ for which $\mathfrak{Re}((x+\bar \lambda y)\circ (x+\bar \lambda y)) =0$, namely $x+\bar \lambda  y\in \epsilon M$, that is $x$ and $y$ project to proportional elements in $V$ and hence, as shown in Sec.\ \ref{jxe}, $\mathfrak{Du}(x\circ y)=0$ and hence $x\circ y=\pm 1$.

Observe that the function $\cos: \mathbb{D}\to \mathbb{D}$ has the form $\cos (a+\epsilon b)=\cos a- \epsilon b \sin a$ thus it sends bijectively $(0,\pi)+\epsilon \mathbb{R}$ to $(-1,1)+\epsilon \mathbb{R}$. Moreover, it sends $\{0,\pi\}$ bijectively to $\{1,-1\}$. This means that there is one and  only one dual number $\Theta\in \{0,(0,\pi)+\epsilon \mathbb{R}, \pi\}$ such that
\[
\cos \Theta =x \circ y.
\]
A standard calculation via Eqs.\ (\ref{bjm}),(\ref{bqd}) shows that $\Theta=\theta+\epsilon d$ where $\theta$ is the angle between the resultants and $d$ is the distance between the axes (but this fact will not be used in the next proofs).

The previous formula shows, with obvious meaning of the notation, $\Theta_{xy}=\Theta_{yx}$, due to the symmetry of the scalar product, and $\Theta_{xy}=\pi-\Theta_{(-x)y}$. Thus $\cos \Theta_{xy}=-\cos \Theta_{(-x)y}$, $\sin \Theta_{xy}=\sin \Theta_{(-x)y}$.

Since every $x\in M\backslash \epsilon M$, can be normalized $x \to \hat x=x/\vert x \vert$ we have for screws rather than oriented lines
\begin{equation} \label{biq}
x \circ y=\vert x\vert \, \vert y\vert \cos \Theta.
\end{equation}
Taking the real part gives the Cauchy-Schwarz inequality for screws.

\subsection{Dual angle and cross product}

Let us consider the cross product of unit screws $x,y\in M\backslash \epsilon M$ with linearly independent resultants $e_1,e_2$. We know that $x\times y$ has resultant $e_1\times e_2$. Let $x\times y=\alpha u $ where $\alpha\in
\mathbb{D}$, $\mathfrak{Re}(\alpha)>0$, and $u$ is a unit screw of direction $e_1\times e_2$.
Now we have
\begin{align*}
0&=x \times y \circ x=\alpha u\circ x,  \\
0&=x \times y \circ y=\alpha u\circ y ,
\end{align*}
 which imply $u\circ x=u\circ y=0$. This proves that the axis of $u$ intersects both the axis of $x$ and $y$ orthogonally.
It remains to find  $\alpha$. Observe that by Eq.\ (\ref{ciq})
\[
\alpha^2=(x \times y) \circ (x \times  y) = (x \circ  x)\,(y \circ y) - (x \circ y)\,(y \circ x)=1-\cos^2 \Theta, \ \ \Rightarrow \ \ \alpha =\sin \Theta.
\]
We conclude that for two unit screws with linearly independent resultants, denoting with $u$ the unit screw of the oriented line orthogonal to both axes and oriented as $e_1\times e_2$,
\[
x \times y= \sin \Theta\, u
\]
so that for general screws with linearly independent resultants
\begin{equation} \label{zjs}
x \times y=\vert x \vert \,\vert y \vert  \sin \Theta\, u.
\end{equation}
Actually this formula holds also for linearly dependent resultants as it is easy to check using Eqs.\ (\ref{bjm}),(\ref{dir}). Since the coefficient in front of $u$ in Eq.\ (\ref{zjs}) is a product of factors with  positive real part we have
\begin{equation} \label{vqx}
 \vert x \times y\vert=\vert x \vert\, \vert y \vert \sin \Theta.
\end{equation}


\section{Manifestations of the transference principle} \label{squa}

In this section we obtain some results in screw theory that could have been obtained by using the transference principle. Indeed, the point of this section is to show that these results admit  proofs that mimic proofs in vector geometry.

\subsection{Linear dependence, proportionality and cross product}

We have shown above that  $z_1,z_2\in M\backslash \epsilon M$ are linearly independent iff  so are the resultants $r_1,r_2\in V$. Since $z_1\times z_2$ has resultant $r_1\times r_2$ this means that $z_1,z_2\in M\backslash \epsilon M$ are linearly independent iff $r_1\times r_2\ne 0$. Let us prove\\

\begin{proposition}
 $z_1,z_2\in M\backslash \epsilon M$ are proportional iff $z_1\times z_2=0$.
\end{proposition}
\begin{proof}
One direction is obvious. For the other direction, we can assume without loss of generality that $z_1\circ z_1=z_2\circ z_2=1$. So suppose $z_1\times z_2=0$, then $r_1 \times r_2=0$, namely the resultants are proportional, and by a previous observation $\mathfrak{Du}(z_1\circ z_2)=0$ which implies that the axes of $z_1$ and $z_2$ intersect and hence are the same, which implies that $z_1,z_2$ are proportional.
\end{proof}

\subsection{Linear dependence, proportionality and mixed product}

Similarly, we have shown above that  $z_1,z_2,z_3\in M\backslash \epsilon M$ are linearly independent iff  so are the resultants $r_1,r_2,r_3\in V$. The real part of $z_1 \times z_2 \circ z_3$ is the scalar product of the resultant of $z_1 \times z_2$ with $r_3$, namely $r_1\times r_2 \cdot r_3$. Thus  $z_1,z_2,z_3\in M\backslash \epsilon M$ are linearly independent iff $\mathfrak{Re}(z_1\times z_2 \cdot z_3)\ne 0$.\\

\begin{proposition} \label{jiq}
Let $z_1,z_2,z_3\in M\backslash \epsilon M$. The mixed product $z_1 \times z_2 \circ z_3$ vanishes iff their axes are "parallel or intersect orthogonally a common line".\\
\end{proposition}

Observe that if the system is in equilibrium, i.e.\ $z_1+z_2+z_3=0$, this case applies.
\begin{proof}
Suppose that $z_1,z_2,z_3\in M\backslash \epsilon M$ have parallel axes, then they have the form (for any chosen $O\in E$) $z_i=\alpha_i (e^O+\epsilon \mathcal{s}_i(O))$, with $\alpha_i\in \mathbb{D}\backslash \epsilon \mathbb{R}$, and $e\in V$ a vector giving the direction of the axes. This easily implies  $z_1 \times z_2 \circ z_3=0$ by application of $e^O$ on the mixed product or by application of $\epsilon$. Another case in which it vanishes is that in which $z_1,z_2,z_3\in M\backslash \epsilon M$ have axes orthogonal to the same line. Indeed, by  Prop.\ \ref{vbb} in this case one screw is a $\mathbb{D}$-linear combination of the other two and hence the triple product vanishes.

Conversely, if $z_1,z_2,z_3\in M\backslash \epsilon M$ are such that $z_1 \times z_2 \circ z_3=0$ then the axes are either parallel or orthogonal to the same line. Without loss of generality we can assume that they are normalized. If  the axes are not parallel, then we can assume that $z_1,z_2$ have linearly independent resultant. Thus  $z_1 \times z_2= \alpha u$ for some $\alpha \in \mathbb{D}\backslash \epsilon \mathbb{R}$, and for some normalized screw $u$. The axis of $u$ is orthogonal to that of $z_1$ and $z_2$ and intersects both. From $u\circ z_3=0$ it follows that also the axis of $z_3$ intersects that of $u$ orthogonally.
\end{proof}

It can be observed that $z_1\times z_2 \cdot z_3=0$ does not imply that one $z$ can be expressed as a $\mathbb{D}$-linear  combination of the other two, e.g.\ the case of three normalized screws with parallel non-complanar axes.\\

\begin{proposition} \label{jip}
Let $z_1,z_2,z_3\in M\backslash \epsilon M$. Suppose that they are pairwise linearly independent (i.e.\ the resultants are pairwise linearly independent/ no pair of axes are parallel), then the following statements are equivalent
\begin{itemize}
\item[(i)]  $z_1 \times z_2 \circ z_3=0$ ,
\item[(ii)] some $z_i$ is a $\mathbb{D}$-linear combination of the other two,
\item[(iii)] the axes intersect orthogonally a common line.
\end{itemize}
If these statements hold then (ii) holds with `every' replacing `some'.
\end{proposition}

\begin{proof}
The equivalence of (i) and (iii) follows from Prop.\ \ref{jiq}, (ii) $ \rightarrow $ (i) is clear, while (iii) $ \rightarrow $ (ii) and the last statement follow from Prop.\ \ref{vbb}.
\end{proof}

\subsection{$\mathbb{R}$-linear dependence}

\begin{proposition}
Let $z_1,z_2\in M\backslash \epsilon M$. They are  $\mathbb{R}$-linear dependent iff they share the same axis and the same pitch.

Let $z_1,z_2.z_3\in M\backslash \epsilon M$.  If they are  $\mathbb{R}$-linear dependent then the axes intersect orthogonally the same line. Moreover, if they are parallel they are coplanar.
\end{proposition}

\begin{proof}
Observe that two screws $z_1,z_2\in M\backslash \epsilon M$ that are $\mathbb{R}$-linear dependent are also linear $\mathbb{D}$-dependent, however the combination $az_1+b z_2=0$ with $a,b\in \mathbb{R}$, implies that they are proportional through a real number factor, thus they share the same axis and furthermore have the same pitch. Clearly, this condition is not only necessary but also sufficient for  $\mathbb{R}$-linear dependence.

Let us consider three screws $z_1,z_2,z_3\in M\backslash \epsilon M$ that are  $\mathbb{R}$-linear dependent. If two of them are $\mathbb{R}$-linear dependent then two axes coincide and the claim is clear. Otherwise all the real coefficients of the vanishing real combination are different from zero and they can be reabsorbed (without  changing the axis) so as to reduce ourselves to the case $z_1+z_2+z_3$ which proves the claim by Prop.\ \ref{vbf}.
\end{proof}

The condition on the orthogonality of the axes implies (Prop.\ \ref{vbb})
$z_3=a z_1+bz_2$ where, however, $a,b\in \mathbb{D}$.
In order for the coefficients to be in the reals the axis of $z_3$ would need to stay in a surface called {\em cylindroid} and its pitch would also need to be constrained \cite{dimentberg68}.\\

\begin{proposition} \label{qid}
Let $z_1,z_2.z_3\in M\backslash \epsilon M$ be sliding vectors (zero pitch) with parallel axes. They are $\mathbb{R}$-linear dependent iff their axes are coplanar.
\end{proposition}

\begin{proof}
Suppose they are $\mathbb{R}$-linear dependent then they are coplanar by the previous result.

Conversely, if they are coplanar rescale one screw so that the total resultant vanishes, then there is still freedom to change the individual resultants to get a vanishing total screw field over one axis and hence everywhere.
\end{proof}

\subsection{Law of cosines}

\begin{theorem}
For any  three proper screws in equilibrium, i.e.\ $x+y+z=0$, $x,y,z \in M\backslash \epsilon M$ we have
\begin{equation} \label{cos}
z^2=x^2+y^2-2 \vert x \vert \, \vert y \vert\cos \alpha_{x y}.
\end{equation}
where   $\alpha:=\pi-\Theta$ (with $\Theta$  the `complementary' angle) so that $x\circ y= \vert x \vert \vert y \vert \cos \Theta_{xy}=- \vert x \vert \vert y \vert \cos \alpha_{xy}$. We have also the analogous formulas obtained permuting $x,y,z$.
\end{theorem} $\empty$

 The fact that they sum up to zero means that they form a `triangle' in $M$.

 \begin{proof}
  Squaring $x+y=-z$  and using Eq.\ (\ref{biq}) gives
\[
x\circ x+y\circ y+2 \vert x \vert\, \vert y \vert\cos \Theta_{x y}=z\circ z
\]
which can also be rewritten in the more standard form of Eq.\ (\ref{cos}).
 \end{proof}
 Equation (\ref{cos}) is the law of cosines for three screws in equilibrium which specializes to the well known law of cosines for triangles (the sides of the triangle can be regarded as  three free vectors whose sum is zero, the standard law of cosines being a result of vector geometry).

\subsection{Law of sines}

\begin{theorem}
For any  three proper screws in equilibrium, i.e.\ $x+y+z=0$, $x,y,z \in M\backslash \epsilon M$ we have
\begin{equation} \label{sin}
\frac{  \sin \alpha_{xy} }{\vert z \vert}=  \frac{\sin \alpha_{yz} }{\vert x\vert } =  \frac{ \sin \alpha_{zx} }{\vert y \vert }.
\end{equation}
\end{theorem} $\empty$

The angle $\alpha$ between screws is  defined above in the law of cosines section.

\begin{proof}
Cross multiplying $x+y=-z$ by $y$ gives
\[
x\times y=y\times z
\]
and using Eq.\ (\ref{vqx})
\[
\vert x \vert\, \vert y \vert  \sin \Theta_{xy}  =\vert y \vert \,\vert z \vert  \sin \Theta_{yz} .
\]
 We conclude
\begin{equation}
\frac{  \sin \Theta_{xy} }{\vert z \vert}=  \frac{\sin \Theta_{yz} }{\vert x \vert } =  \frac{ \sin \Theta_{zx} }{\vert y \vert}
\end{equation}
where the last equality follows from the fact that the starting equation $x+y+z=0$ is invariant under permutations of $x,y,z$. This is the law of sines for screws, which can also be written in the more standard form of Eq.\ (\ref{sin}).
\end{proof}


Observe that, denoting with $2R (\in \mathbb{D})$ the constant appearing on the right-hand side of Eq.\ (\ref{sin}), we have
\[
4 R^2\vert x \vert^2\,  \vert y\vert^2 \, \vert z\vert^2=\vert x \vert^2 \, \vert y\vert^2-(x \circ y)^2=\vert y \vert^2 \, \vert z\vert^2-(y \circ z)^2=\vert z \vert^2 \, \vert x\vert^2-(z \circ x)^2 .
\]

\subsection{Sum of angles in a triangle}

To be precise, the  following result is not a complete manifestation of our version of the transference principle as the proof in the vector geometry case is complete already at Eq.\ (\ref{bkq}). It has to be completed to show that the sum of the angles has no dual part.\\


\begin{proposition}
For three screws in equilibrium, i.e.\ $x+y+z=0$, $x,y,z \in M\backslash \epsilon M$, we have
\begin{equation}
\alpha_{xy}+\alpha_{yz}+\alpha_{zx}=\pi.
\end{equation}
\end{proposition}

\begin{proof}
Notice that every equation  $f(d_1,\cdots,d_n)=0$ with $f$ real analytic and $d_i$ dual numbers, in order to hold, needs only to hold in its real part \cite{dimentberg68}. This is because the derivative of a zero function is zero. It follows that the next identity with $A,B,C\in \mathbb{D}$ holds, as it holds for real angles,
\[
\cos(A+B+C)=\cos A\cos B\cos C-\sum_{\textrm{cyclic}} \sin A\sin B \cos C,
\]
where the (cyclic) sum is over all the even permutations of $(A,B,C)$.
Alternatively, it can be proved by expanding $\cos (a+\epsilon b)=\cos a-\epsilon b \sin a$, and similarly for the sine.
Thus
\begin{align*}
\cos (\alpha_{xy}+\alpha_{yz}+\alpha_{zx})&=\cos \alpha_{xy}\,\cos  \alpha_{yz}\,\cos \alpha_{zx} -\sum_{\textrm{cyclic}} \sin \alpha_{xy}\,\sin \alpha_{yz} \,\cos \alpha_{zx}\\
&= -\frac{x \circ y}{\vert x\vert \, \vert y\vert} \, \frac{y \circ z}{\vert y\vert \, \vert z\vert} \,  \frac{z \circ x}{\vert z\vert\,  \vert x\vert} -\sum_{\textrm{cyclic}} \frac{\vert z \vert }{\vert x \vert }\,\sin^2 \alpha_{yz} \,\cos \alpha_{zx}
\end{align*}
where in the last term we used the law of sines. Using
\[
\sin^2 \alpha_{yz}=\left[\vert y \vert ^2 \, \vert z \vert ^2-(y\circ z)^2\right]/\vert y \vert ^2 \, \vert z \vert ^2=4 R^2 \vert x\vert^2
\]
we get
\begin{align*}
\cos (\alpha_{xy}+\alpha_{yz}+\alpha_{zx})&= -\frac{(x \circ y)\, (y\circ z) \, (z \circ x)}{\vert x\vert^2 \, \vert y\vert^2 \, \vert z \vert^2}  -\sum_{\textrm{cyclic}} 4R^2 \vert z \vert \vert x \vert  \,\cos \alpha_{zx}\\
&= -\frac{(x \circ y)\,(y\circ z)\, (z \circ x)}{\vert x\vert^2 \, \vert y\vert^2 \, \vert z \vert^2} + 4R^2 (x \circ y+y \circ z+z \circ x) ,
\end{align*}
and using $z=-x-y$ we arrive at
\begin{align*}
&\cos (\alpha_{xy}\!+\!\alpha_{yz}\!+\!\alpha_{zx})=\! -\frac{(x \circ y)\,(x \circ y\!+ \! y^2)\,(x\circ y\! + \! x^2) }{\vert x\vert^2 \, \vert y\vert^2 \, \vert z \vert^2} - 4R^2 (x \circ y\!+\!x^2\!+\!y^2)\\
&= \!-\frac{(x \circ y)\,(x \circ y\!+\! y^2)\,( x\circ y \!+\! x^2)+( x^2  y^2-(x \circ y)^2)\,(x \circ y\!+\!x^2\!+\!y^2)}{\vert x\vert^2 \, \vert y\vert^2 \, \vert z \vert^2}=-1 .
\end{align*}

Note that $0\le\mathfrak{Re}(\alpha_{xy}+\alpha_{yz}+\alpha_{zx})\le 3\pi$ and they cannot all be $\pi$ otherwise the resultants would be all aligned and could not sum to zero unless they are all zero, which is excluded.  We arrive at the preliminary result
\begin{equation} \label{bkq}
\alpha_{xy}+\alpha_{yz}+\alpha_{zx}-\pi\in \epsilon \mathbb{R}.
\end{equation}

Now we use
\begin{equation} \label{voa}
\sin(A+B+C)=-\sin A\sin B\sin C+\sum_{\textrm{cyclic}} \cos A\cos B \sin C.
\end{equation}
\begin{align*}
&\sin (\alpha_{xy}+\alpha_{yz}+\alpha_{zx}) = -(2R)^3\vert x\vert \, \vert y\vert \, \vert z\vert+ k \sum_{cycl}  \left(\frac{-x \circ y}{\vert x\vert \, \vert y\vert}\right)\left( \frac{-y \circ z}{\vert y\vert \, \vert z\vert} \right) \vert y\vert\\
&=\!-(2R)^3\vert x\vert \, \vert y\vert \, \vert z\vert+ \frac{2R}{\vert x\vert \, \vert y\vert \, \vert z\vert} [(x\circ y)\,(y\circ z)\!+\!(y\circ z)\,(z\circ x)\!+\!(z\circ x)\,(x\circ y)]\\
&=\!\frac{2R}{\vert x\vert\,  \vert y\vert \, \vert z\vert}\left\{-4R^2 x^2 y^2 z^2+[(x\circ y)\,(y\circ z)\!+\!(y\circ z)\,(z\circ x)\!+\!(z\circ x)\,(x\circ y)]\right\}
\end{align*}
Using $4R^2 x^2  y^2  z^2=x^2 y^2-(x \circ y)^2$ and $z=-x-y$ we arrive at
\begin{align*}
&\sin (\alpha_{xy}+\alpha_{yz}+\alpha_{zx})=\frac{2R}{\vert x\vert \, \vert y\vert\, \vert z\vert} \Big\{-x^2 y^2+(x \circ y)^2 \\
&-\left[(x\circ y)\, (y\circ x\!+\!y^2)-(y \circ x\!+\!y^2)\,(y\circ x\!+\!x^2)+(y\circ x\!+\!x^2)\,(x\circ y)\right]\Big\}=0
\end{align*}
since $\sin (\pi+\epsilon b)=-\epsilon b$, this proves
$
\alpha_{xy}+\alpha_{yz}+\alpha_{zx}=\pi.
$
Alternatively, one could have used Prop.\ \ref{jiq} to show first that the three translations associated to the three angles are actually aligned (at least when the resultants are linearly independent), and then complete the proof from there.
\end{proof}

\subsection{Petersen-Morley theorem}

In this section we review a fairly elementary  proof of the Petersen-Morley theorem \cite{todd36}   and show that it is in fact the analog of a theorem in vector geometry. The Petersen-Morley theorem is a statement about lines in space in  generic position, but it is most easily expressed in terms of screws as done below, Theorem \ref{kke}, where the line is the axis of the screw.

We recall that two screws $x,y \in M\backslash \epsilon M$ are normal if $x\circ  y=0$ which is equivalent to the fact that their axes intersect orthogonally. We denote with $n_{xy}$ the line normal to both axes.
Observe that for non-proportional resultants, $x \times y$ has axis which is the normal to both the axes of $x$ and $y$. Moreover, by Prop.\ \ref{jiq} three screws in equilibrium $a+b+c=0$ with linearly independent resultants have necessarily a common normal $n:=n_{ab}=n_{bc}=n_{ca}$.

Observe that these results hold also for vectors, i.e.\ elements of $(V,\cdot, \textrm{or})$, with obvious adjustments in the meaning of the terminology.

The Petersen-Morley theorem is just a consequence of the Jacobi identity
\begin{equation}
x \times (y\times z)+z \times (x\times y)+y \times (z\times x)=0,
\end{equation}
namely it follows from the fact that the screws
\[
a:=x \times (y\times z), \quad b:=z \times (x\times y), \quad c:=y \times (z\times x)
\]
are in equilibrium. We shall assume that $x,y,z$ are {\em generic} meaning that none of their cross products nor $a,b,c$ vanish. This is a rather weak condition of genericity as there are instances in which the axes of $x,y,z$ could even be coplanar (see below).

Let $n_{(xy)z}$ be the line normal to both $n_{xy}$ and the axis of $z$.\\

\begin{theorem}[Petersen-Morley] \label{kke}
If $x,y,z\in M\backslash \epsilon M$ are generic then $n_{(xy)z}$, $n_{(zx)y}$, $n_{(yz)x}$ have a common normal.
\end{theorem}

\begin{proof}
The axis of $y\times z$ is $n_{yz}$ thus that of $a$ is $n_{(yz)x}$. Since $a+b+c=0$ the theorem follows from Prop.\ \ref{vbf}.
\end{proof}

Obviously this is actually the translation of a result in vector geometry, still following from the Jacobi identity in precisely the same way, the axis of a vector being read as its span, and $x,y,z\in V\backslash 0$.


\subsection{Thales' theorem}
In geometry, Thales's theorem states that if $A$, $B$, and $C$ are distinct points on a circle where the line $AB$ is a diameter, the angle $BCA$ is a right angle. We are going to obtain a screw theory analog.

Let $r\in \mathbb{D}$, $\mathfrak{Re}(r)>0$. The sphere  of radius $r$ is the locus $S(r):=\{z\in M: \vert z\vert=r\}$.
Observe that the sphere condition implies that $S(r)\cap \epsilon M=\emptyset$ but does not place any restriction on the screw axis, it just fixes  the modulus of the resultant and the pitch.\\

\begin{theorem}
Let $x,y,z\in S(r)$. If $x=-y$ then $(y-z)\circ (z-x)=0$.
\end{theorem}

\begin{proof}
Indeed $-(y-z)\circ (z-x)=(z+x)\circ (z-x)=z^2-x^2=r^2-r^2=0$.
\end{proof}

\begin{remark}
This trivial result has some unexpected consequences. We give a kinematical interpretation. We assume the reader to be familiar with the fact that the rigid motion at each instant is described by a screw field of velocities, and that the composition of motions corresponds to the sum of the associated kinematical screws \cite{minguzzi12}.

Thus, consider three frames/observers $A,B,C$ in motion with respect to an inertial frame.
Suppose that all frames have equal modulus of the angular velocity and equal modulus of the translational velocity (invariant $\vert v_\parallel\vert$,   which follows from constant pitch). Note that the instantaneous axes of rotation of the frames do not need to pass through the same point.

Furthermore, suppose that two observers, say $A$ and $B$, have opposite kinematical screws (opposite motions) in the inertial frame, hence sharing the same screw axis.
The theorem states that with respect to the third observer ($C$) frames $A$ and $B$ are observed having instantaneous axes of rotation that intersect orthogonally.

\end{remark}

\section{Conclusions}

We have shown that screw theory can be developed from a $\mathbb{D}$-module of rank 3 endowed with a scalar product and an orientation $(M,\circ, OR)$ (a triple that we termed {\em $\mathbb{D}$-module geometry}). The Euclidean space $E$ is constructed out of certain real $\langle,\rangle$-null planes and each element of the module gets represented by a screw field over such Euclidean space. In short we proved that $\mathbb{D}$-module geometry leads to Euclidean   geometry where the elements of the  $\mathbb{D}$-module get identified with screw fields over $E$. This result shows that the introduction of the dual unit $\epsilon$, $\epsilon^2=0$, leads from vector geometry to Euclidean geometry, that is, it brings us to the affine space while preserving the linear algebraic structure.

The introduced $\mathbb{D}$-module geometry  formalism for dealing with screw theory preserves the effectiveness of screw calculus while removing any reference to reduction points. In this formalism dual vectors can be completely avoided, as one can advantageously learn how concepts from Euclidean geometry correspond to concepts of $\mathbb{D}$-module geometry. For instance, orthonormal frames in $E$  correspond to orthonormal bases on $M$.
This new formalism  has some  advantages, for instance, it clarifies the emergence of the principle of transference as a consequence of the  formal similarity between  $\mathbb{D}$-module geometry $(M,\circ, OR)$ and vector geometry $(V,\cdot, or)$.

Applications of the $\mathbb{D}$-module formalism to mechanics, particularly robotics, and geometry will be given in a next work.

\end{document}